\documentclass[A4paper,onecolumn,draftcls,12pt,conference]{IEEEtran}
\usepackage{graphics,graphicx}
\usepackage{amssymb,amsmath}
\usepackage{amsmath}
\usepackage{comment}
\usepackage{url}
\usepackage[utf8]{inputenc}
\usepackage{algorithm}
\usepackage{hyperref}
\usepackage{indentfirst}
\usepackage{setspace}
\usepackage{cases}

\usepackage{booktabs}
\usepackage{multicol}
\usepackage{mathtools}

\usepackage[noend]{algpseudocode}

\usepackage{amsthm}

\newtheorem{theorem}{Theorem}

\newtheorem{lemma}{Lemma}

\theoremstyle{remark}
\newtheorem{remark}{Remark}

\usepackage{stfloats}

\usepackage{setspace}
\usepackage{lettrine}

\usepackage{suffix}
\usepackage{mathtools}

\DeclarePairedDelimiterX\MeijerM[3]{\lparen}{\rparen}%
{\begin{smallmatrix}#1 \\ #2\end{smallmatrix}\delimsize\vert\,#3}

\newcommand\MeijerG[8][]{%
  G^{\,#2,#3}_{#4,#5}\MeijerM[#1]{#6}{#7}{#8}}

\WithSuffix\newcommand\MeijerG*[7]{%
  G^{\,#1,#2}_{#3,#4}\MeijerM*{#5}{#6}{#7}}

\newcommand{\ui}{\textsc{I}}
\newcommand{\uo}{\textsc{O}}

\usepackage{booktabs}
\usepackage[normalem]{ulem}
\usepackage[nolist]{acronym}

\usepackage{tikz}
\definecolor{gold}{rgb}{0.85,.66,0}
\definecolor{cian}{rgb}{.02,.7,.95}
\definecolor{ppp}{rgb}{.7,.3,.82}

\newcommand{\w}{\textcolor{cian}} 

\newcommand{\diag}{{\rm diag}} 

\def\BibTeX{{\rm B\kern-.05em{\sc i\kern-.025em b}\kern-.08em
    T\kern-.1667em\lower.7ex\hbox{E}\kern-.125emX}}

\begin{document}
\title{NOMA-aided double RIS under Nakagami-$m$ fading: Channel and System Modelling}
\author{\IEEEauthorblockN{Wilson de Souza Junior}
\IEEEauthorblockA{\textit{Department of  Electrical Engineering (DEEL)} \\
\textit{State University of Londrina  (UEL).}\\
86057-970, Londrina, PR, Brazil. \\
 wilsoonjr98@gmail.com \vspace{-5mm}}
\and
\IEEEauthorblockN{ {Taufik Abrão}}
\IEEEauthorblockA{\textit{Department of  Electrical Engineering (DEEL)} \\
\textit{State University of Londrina  (UEL).}\\
86057-970, Londrina, PR, Brazil. \\ taufik@uel.br
\vspace{-5mm}}
}

\maketitle

\begin{abstract}
We investigate the downlink outage performance of double-RIS-aided non-orthogonal multiple access (NOMA), where a near-BS and a near-users RISs setup are deployed. To extend the coverage to 360 degrees, we deploy a simultaneously transmitting and reflecting RIS (STAR-RIS) structure aiming to improve communication reliability for both indoor and outdoor users. New channel statistics for the end-to-end channel with Nakagami-$m$ considering both the conventional-RIS and the STAR-RIS antenna elements features are derived using the moment-matching (MM) technique. The numerical results reveal that the double-RIS setup can outperform the single-RIS setups when the number of elements of STAR-RIS ($R_S$) and conventional RIS ($R_C$) is suitably adjusted. Moreover, when the link between the base station and the near-user RIS is in good condition, the double-RIS setup is outperformed by the single-RIS setup. Finally, the proposed analytical equations reveal to be very accurate under different channel and system configurations.
\end{abstract} 
\acresetall
\smallskip
\noindent 


\section{Introduction}
\lettrine[findent=2pt]{\textbf{T}}{}he upcoming beyond fifth generation (B5G) wireless networks aim to provide high connectivity and different services for the most different kinds of users. Reconfigurable intelligent surface (RIS) has gained much space both in the research community and industries, and it is envisioned as a promising candidate to meet the stringent demands of future networks, mainly due to its ability to create a configurable propagation environment, enhancing metrics such as Spectral Efficiency (SE) \cite{9863732}, Energy Efficiency (EE) \cite{10015822}, as well as indicators as Quality of Service (QoS). However, depending on the place where the conventional RISs (C-RISs) are installed, users located in specific positions (blind spots) will not be served. Recently a new proposed architecture named simultaneously transmitting and reflecting RIS (STAR-RIS) can overcome this limitation by extending the coverage of communication {by} $180^{\circ}$ to $360^{\circ}$. Besides, it has been shown that the RIS should be installed closer to the receiver to achieve the best performance gain\cite{9500188}; thus, in order to improve the potential of this technology, double-RISs have been attracting much attention due to the high potential by combining the RIS deployment  near to the transmitter and receiver.  

Furthermore, with the exponential growth of connected devices in the network, schemes able to support massive connectivity are demanded, hence, non-orthogonal multiple access (NOMA) is considered a promising multiple access technique for B5G since it can provide access for a higher number of users, improving the system spectral efficiency. 
Recently, there is an effort to predict the behavior of the RIS-aided channels; hence, analytical expressions for {\it outage probability} (OP) and {\it ergodic capacity} (EC) considering different multiple access systems have been proposed in \cite{cheng2021downlink, tahir2020analysis, han2022double, yue2022simultaneously}. These analytical approaches can be useful to assess the system's reliability and understand how the RIS-aided system scales according to some parameters.
In \cite{cheng2021downlink}, the authors derived analytical expressions for OP and EC for NOMA and OMA in downlink and uplink networks. In \cite{tahir2020analysis} the authors derived analytical equations for OP in a NOMA-aided single RIS scenario with random and optimal phase-shift design through the statistical moment-matching (MM) technique. In \cite{han2022double}, the authors proposed a method to optimize the phase-shift of double-RIS-aided MIMO communication setup to maximize the system capacity. \cite{yue2022simultaneously} investigated the performance of STAR-RIS assisted NOMA networks over Rician fading channels, and analytical equations for OP are derived.

To the best of the authors’ knowledge, the analytical channel prediction behavior of a double-RIS (C-RIS + STAR-RIS)-aided NOMA has not yet been considered in the literature. Hence, motivated by the research gap considering this scenario we investigated the closed-form expressions for OP and EC; besides, we characterized conditions in which double-RIS configurations can be indeed profitable. The major contributions of this work are threefold: 
\textit{\textbf{i}}) we analyze a downlink double-RIS system, where a conventional RIS and a STAR-RIS are deployed to serve an indoor and an outdoor user. Through the MM technique, we statically characterize the end-to-end channel for this new scenario where all the communications links are considered as Nakagami-$m$;
\textit{\textbf{ii}}) the OP and EC are derived in simple and accurate closed-form expressions and calculated for different parameters configurations of the double-RIS setup, as the number of elements in each RIS, as well as the $m$ parameter of Nakagami distribution (LoS and NLoS impact). The developed analytical results are confirmed by Monte-Carlo simulations (MCs);
\textit{\textbf{iii}}) we expose in which conditions and configurations the double-RIS setup can outperform the single-RIS setup.

\noindent\textit{Notations:} Scalars, vectors, and matrices are denoted by the lower-case, bold-face lower-case, and bold-face upper-case letters, respectively; $\mathbb{C}^{N\times M}$ denotes the space of $N\times M$ complex matrices; $|\cdot|$ denote the absolute operator; $\diag(\cdot)$ denotes the diagonal operator; $\Gamma(\cdot)$ is the gamma function; $\gamma(\cdot)$ is the lower incomplete gamma function; $\mathbb{E}[\cdot]$ and $\mathbb{V}[\cdot]$ denotes the statistical expectation and variance operator; $\mathcal{G}(k,\theta)$ denotes the
distribution of a Gamma random variable with shape and scale parameter $k$ and $\theta$ respectively; $\mathcal{G}\mathcal{G}(\alpha,\beta,\gamma,c)$ denotes an r.v. whose distribution is generalized-Gamma where the parameters are $\alpha,\beta,\gamma$ and $c$; $Nakagami(m,\Omega)$ denotes the
distribution of a Nakagami-$m$ random variable with shape and spread parameter $m$ and $\Omega$ respectively;

\section{System Model}\label{sec:sys_model}
As shown in Fig. \ref{fig:RIScell}, we consider a downlink double-RIS-aided two-user communication system, where the base station (BS), denoted by $S$, and the users are equipped with one single antenna. To extend the coverage to 360 degrees, we consider a STAR-RIS deployed in the building facade which can enhance the communication of indoor user ($U_\ui$) and outdoor user ($U_\uo$) simultaneously. We assume the $U_\ui$ and $U_\uo$ as the strongest and weakest users, respectively. The RISs are located near $S$ (conventional RIS referred to as $R_C$) and near users (STAR-RIS referred to as $R_S$). Let us denote as $N$ the total number of passive elements for the two RIS, where $R_C$ and $R_S$ consist of $N_C = |\mathcal{N}_C|$ and $N_S = |\mathcal{N}_S|$ elements respectively, with the set of RIS elements defined as $\mathcal{N}_C = \{1,\dots N_C \}$, $\mathcal{N}_S = \{1,\dots N_S \}$ and $N_C+N_S=N$. We define the {\it split factor} as $\eta \triangleq \frac{N_C}{N}$. Furthermore, we consider that $U_\ui$ and $U_\uo$ are paired to perform NOMA.

The $R_C$ can be configured aiming to improve the communication through its phase-shift matrix given as  $\boldsymbol{\Phi}_C = \diag([\nu_1e^{j \theta_1},\dots,\nu_{N_C}e^{j \theta_{N_C}}])$, where $\theta_n \in [0,2\pi]$ and $\nu_n \in (0,1]$ are the phase-shift and the amplitude coefficient applied at the $n$-th element of $R_C$, respectively. Here, we adopt ideally $\nu_n =1$, $\forall n \in\{1,\dots,N_C\}$. Besides, concerning the STAR-RIS ($R_{S}$), its phase-shift matrix can be defined as $\boldsymbol{\Phi}_S^i = \diag([\xi_{i,1}e^{j \varphi_{i,1}},\dots,\xi_{i,N_S}e^{j \varphi_{i,N_S}}])$ with $i \in \{I,O\}$, where $\xi_{\ui,n}$ ($\varphi_{\ui,n}$) and $\xi_{\uo,n}$ ($\varphi_{\uo,n}$) denotes the amplitude coefficient (phase-shift) for refraction and reflection of $n$-th element respectively. 

This paper focuses on the energy splitting (ES) protocol, where the STAR-RIS is assumed to operate in transmitting and reflecting mode simultaneously by splitting the total radiation energy into the refraction and reflection signals. By obeying the energy conservation law, the sum of the energies of the transmitted and reflected signals should be equal to the incident signal's energy, $i.e.$, $\xi_{\ui,n}+\xi_{\uo,n}=1 \; \forall n \in \mathcal{N}_S$. Without loss of generality we adopt $\xi_{i,n}=\xi=0.5$, $\forall n \in \mathcal{N}_S$ and $\forall i \in \{I,O\}$.

\subsection{Channel Model} \label{ss:channelmodel}

In order to systematically characterize our proposed system model, let us denote $\mathbf{f} = [f_1,\dots,f_{N_C}]$ $\in \mathbb{C}^{1\times N_C}$, $\mathbf{D} = [\mathbf{d}_1,\dots,\mathbf{d}_{N_S}]$ $\in \mathbb{C}^{N_C\times N_S}$, $\mathbf{u}^{i}=[u^i_1,\dots,u^i_{N_S}]^T$ $\in \mathbb{C}^{N_S\times 1}$ $\forall i \in \{\ui,\uo\}$, $\mathbf{g} =[g_1,\dots,g_{N_C}]^T$ $\in \mathbb{C}^{N_C\times 1}$ and $\mathbf{t}=[t_1,\dots,t_{N_S}]$ $\in \mathbb{C}^{1 \times N_S}$ as the channel between the {$S$} $\rightarrow$ $R_C$, $R_C$ $\rightarrow$ $R_S$, $R_S$ $\rightarrow$ $U_{i}$, $R_C$ $\rightarrow$ $U_\uo$\footnote{In this work, we 
do not consider the link between the $R_C$ $\rightarrow$ $U_\ui$.} and $S$ $\rightarrow$ $R_S$ respectively. We consider that the absolute value of all elements of $\mathbf{f}, \mathbf{D}, \mathbf{u}^i, \mathbf{g}$ and $\mathbf{t}$ follow a Nakagami-$m$ distribution with parameters $(m_f,\Omega_f)$, $(m_D,\Omega_D)$, $(m_{U_i},\Omega_{U_i})$, $(m_g,\Omega_g)$ and $(m_t,\Omega_t)$, respectively. Since the channel estimation problem is out of the scope of this manuscript, we assume that perfect channel state information (CSI) is available at $S$.

\begin{figure}[htbp!]
\vspace{-3mm}
\centering
\includegraphics[width=.8\linewidth]{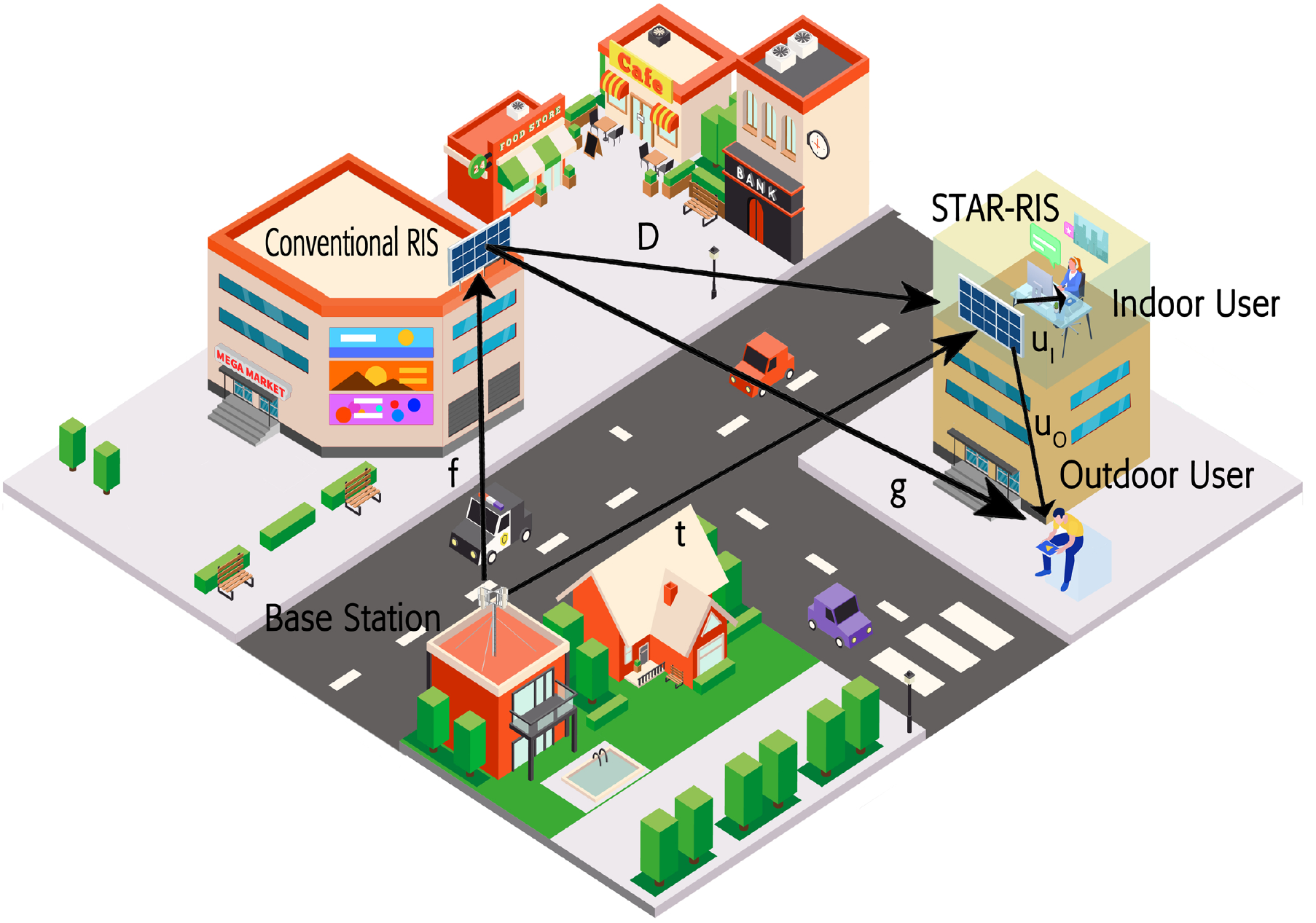}
\vspace{-3mm}
\caption{Application scenario for double RIS-aided mobile communication.}
\label{fig:RIScell}
\end{figure}

\vspace{-3mm}
\section{New Statistics of the end-to-end Channel} \label{ss:stastistcs}

In this section, we intend to derive new channel statistics for the proposed double-RIS-aided network, which will be useful to compute the OP and EC in the following subsections. 

As illustrated in Fig. \ref{fig:RIScell} we consider the double reflection link as well as the single reflection link from $S$ to $R_S$ towards the $U_\ui$ and $U_\uo$ and from the $R_C$ to the $U_\uo$, therefore, one can write the cascaded channels of both user as follows:
\begin{IEEEeqnarray}{rCl}
&&h_I = \underbrace{\sqrt{\beta_{t}\beta_{U_\ui}}\mathbf{t} \boldsymbol{\Phi}_S^I \mathbf{u}_I}_{\text{single-reflection}} + \underbrace{\sqrt{\beta_{f}\beta_{D}\beta_{U_\ui}}\mathbf{f}
    \boldsymbol{\Phi}_C \mathbf{D} \boldsymbol{\Phi}_S^{I} \mathbf{u}_I}_{\text{double-reflection}}, \label{eq:hi}
    \\
    &&h_O = \small \underbrace{\sqrt{\beta_{f}\beta_{g}}\mathbf{f} \boldsymbol{\Phi}_C \mathbf{g}}_{\text{single-reflection}} + \underbrace{\sqrt{\beta_{t}\beta_{U_\uo}}\mathbf{t} \boldsymbol{\Phi}_S^O \mathbf{u}_O}_{\text{single-reflection}} \nonumber \\ && \hspace{3cm} + \underbrace{\sqrt{\beta_{f}\beta_{D}\beta_{U_\uo}}\mathbf{f}
    \boldsymbol{\Phi}_C \mathbf{D} \boldsymbol{\Phi}_S^{O} \mathbf{u}_O}_{\text{double-reflection}}, \label{eq:ho}
\end{IEEEeqnarray}
where $\beta_v = \left(\frac{d_0}{d_v}\right)^{\alpha_v}$ denotes the {\it large-scale fading} gain (path-loss), $d_0$ is the reference distance, $d_v$ is the distance of link $v$, and $\alpha_v$ is the path-loss coefficient, $\forall v \in \{f,g,t,D,{u_\ui},{u_\uo}\}$.

Notice that in order to maximize the single-reflection link for the $U_\uo$ user, the optimal phase-shift in $R_C$ can be applied setting $\theta_n=-\left(\angle f_n + \angle g_n \right)\; \forall n\in \mathcal{N}_C$; besides, aiming to explore the potential of the deployed $R_S$, we consider the case where it is configured to maximize the received power from the double-reflection link of both users, $i.e.$, the phase-shifts of the $R_S$ are set as $\varphi_{i,n} = -\left(\angle([\mathbf{f}\boldsymbol{\Phi}_C\mathbf{D}]_n) + \angle u_{i,n}\right)\; \forall n\in \mathcal{N}_S$, therefore, the signal of $U_\ui$ and $U_\uo$ will be coherently combined by the $R_S$. Thus, by setting the coherent phase shift design in both RISs, the following lemma is proposed.

\begin{lemma} \label{lemma:Channel}
 $|h_\ui| \overset{approx}{\sim} \mathcal{G}\left(k_{h_\ui},\sqrt{\beta_{U_\ui}} 
 \theta_{h_\ui}\right)$ whose scale and shape parameters are given by Eq. \eqref{eq:parameters_hi}. Moreover, the distribution $|h_O| \overset{approx}{\sim} \mathcal{G}\left(k_{h_O}, \theta_{h_O} \right)$ whose scale and shape parameters are given by Eq. \eqref{eq:parameters_ho}, where  $ \mu_v = \frac{\Gamma(m_v+1/2)\sqrt{\Omega_v}}{\Gamma(m_v)\sqrt{m_v}}$, $\forall v \in \{f,g,t,D,{u_\ui},{u_\uo}\}$. 

\begin{equation}
\begin{cases}
    k_{h_\ui}  = \frac{N_S \mu_{U_\ui}^2 \frac{\pi}{4}}{ \left (\Omega_{{U_\ui}} -  \mu^2_{U_\ui} \frac{\pi}{4} \right)}, 
    \\
    \\
    \theta_{h_\ui} = \frac{\xi \sqrt{\beta_{U_\ui}}\sqrt{(\beta_t+\beta_f\beta_D N_C)}( \Omega_{U_\ui} - \mu^2_{U_\ui} \frac{\pi}{4} )}{ \sqrt{\frac{\pi}{4}} \mu_{U_\ui}}, \label{eq:parameters_hi}  
\end{cases}
\end{equation}

\begin{equation}
\scriptstyle{
    \begin{cases}
     \hspace{-0.1cm} k_{h_O} \hspace{-0.15cm} = \hspace{-0.1cm} \frac{ \left( \sqrt{\beta_{U_\uo}} \xi N_S \mu_{U_\uo}\sqrt{\frac{\pi(\beta_t+\beta_f\beta_D N_C)}{4}} + \sqrt{\beta_f \beta_g} N_C \mu_f \mu_g \right)^2}{\hspace{-0.12cm} \beta_f \beta_g N_C \hspace{-0.05cm} \left(\Omega_f\Omega_g \hspace{-0.05cm}-\hspace{-0.05cm} \mu_f^2 \mu_g^2 \right) + \xi^2 \beta_{U_\uo} N_S (\beta_t+\beta_f\beta_D N_C) \hspace{-0.08cm}\left (\hspace{-0.1cm}\frac{4\Omega_{U_\uo} \hspace{-0.1cm}-\mu_{U_\uo}^2\pi}{4} \hspace{-0.1cm} \right) },
    \\
    \\
     \hspace{-0.1cm} \theta_{h_O}  = \frac{ \sqrt{\beta_{U_\uo}}  \xi N_S \mu_{U_\uo} \sqrt{\frac{\pi(\beta_t+\beta_f\beta_D N_C)}{4}} + \sqrt{\beta_f\beta_g} N_C \mu_f \mu_g  }{k_{h_O}}.\label{eq:parameters_ho}
    \end{cases}
    }
\end{equation}

\begin{proof}
The proof can be found in Appendix \ref{app:a}.
\end{proof}
\end{lemma}

\begin{remark}
If only the double-reflection link is considered,  $\beta_t=\beta_g=0$, besides, if only the single-reflection link is considered, we have that $\beta_D=0$.
\end{remark}

Besides, the probability density function (PDF) and cumulative distribution function (CDF) of $|h_{i}|$ can be given as \cite{papoulis1990probability}
\begin{equation} \label{eq:gammapdf}
    f_{h_i}(x) = \dfrac{1}{\Gamma(k_{h_i})\theta_{h_i}^{k_{h_i}}} x^{k_{h_i}-1} e^{-\frac{x}{\theta_{h_i}}},  
\end{equation}

\begin{equation} \label{eq:gammacdf}
    F_{h_i}(x) = \dfrac{1}{\Gamma(k_{h_i})} \gamma\left(k_{h_i},\dfrac{x}{\theta_{h_i}} \right).    
\end{equation}

%


\vspace{2mm}
\section{Outage Probability and Ergodic Capacity}

We derive new OP closed-form expressions for double-RIS systems aided by STAR-RIS topology.

\subsection{Signal Model of NOMA under Power Domain}

For downlink communications, the $S$ transmits the signal $x = \sqrt{P}\left(\sqrt{{\lambda_I}}s_I + \sqrt{{\lambda_O}}s_O\right)$, where $P$ denotes the transmit power of the $S$, $s_I$ and $s_O$ denote the transmitted signals to $U_\ui$ and $U_\uo$, respectively, and ${\lambda_I}$ and ${\lambda_O}$ denote the power allocation coefficients for
$U_\ui$ and $U_\uo$, respectively, with ${\lambda_I} + {\lambda_O} = 1$, and {${\lambda_I} < {\lambda_O}$}. {On} the indoor and outdoor user side, the received signal is given by
\begin{equation}
y_i = h_i x + z_i, \qquad i \in \{I,O\},
\end{equation}
where the background noise $z_i \sim \mathcal{CN}(0, \sigma^2)$.

At $U_\ui$, the signal of $U_\uo$ is detected first, and the corresponding signal-to-interference-plus-noise ratio (SINR) is given by
\begin{equation}
\gamma_{\textsc{o}\rightarrow \textsc{i}} = \frac{\rho |h_I|^2 {\lambda_O}}{\rho |h_I|^2{\lambda_I}+1},
\end{equation}
where $\rho = \frac{P}{\sigma^2}$ is the transmit SNR under the AWGN channel. Besides, assuming perfect signal reconstruction and canceling in the successive interference cancellation (SIC) step, the SNR at the $U_\ui$ is simply given by:
\begin{equation} \label{eq:gammai}
\gamma_{\ui} = \rho |h_\ui|^2 {\lambda_I}.
\end{equation}
At $U_\uo$, its signal is decoded directly by regarding $U_\ui$’s signal as
interference, resulting in the following SINR
\begin{equation} \label{eq:gammao}
\gamma_{\uo} = \frac{\rho|h_\uo|^2{\lambda_O}}{\rho|h_\uo|^2{\lambda_I}+1}.
\end{equation}

\subsection{Outage Probability at Indoor User}

According to the principle of NOMA, the outage behavior for the $U_\ui$ occurs when it cannot effectively detect the user $U_\uo$' message and its message consequently or when the $U_\uo$' message is decoded successfully but an error occurs in the SIC process to decode its own message. Mathematically it can be formulated by the union of the following events:
\begin{IEEEeqnarray}{rCl}
     P^{U_\ui}_{\rm out} &=& 1 - {\rm Pr}(\gamma_{{\textsc{i}}\rightarrow {\textsc{o}}} > \Bar{\gamma}_{{\rm th}_\uo}, \gamma_{\textsc{i}} > \Bar{\gamma}_{{\rm th}_\ui} ) \nonumber 
     \\
    &=& 1 - {\rm Pr}\left( |h_\ui| > \sqrt{\dfrac{\Bar{\gamma}_{{\rm th}_\uo}}{\rho({\lambda_O}-{\lambda_I} \Bar{\gamma}_{{\rm th}_\uo})}}, |h_\ui| > \sqrt{\dfrac{\Bar{\gamma}_{{\rm th}_\ui}}{{\lambda_I}\rho}} \right), \nonumber
\end{IEEEeqnarray}
where $\Bar{\gamma}_{{\rm th}_\uo}$ and $\Bar{\gamma}_{{\rm th}_\ui}$ are the SINR threshold for $U_\uo$ and $U_\ui$, respectively.  Clearly, if ${\lambda_O}\leq{\lambda_I}\Bar{\gamma}_{{\rm th}_\uo}$, $P^{U_\ui}_{\rm out}=1$, otherwise by utilizing Eq. \eqref{eq:gammacdf} the OP of $U_\ui$ can be given as
\begin{equation} 
P^{U_\ui}_{\rm out} = \dfrac{1}{\Gamma(k_{h_\ui})} \gamma\left(\hspace{-1.2mm}k_{h_\ui},\dfrac{\max\left\{\sqrt{\dfrac{\Bar{\gamma}_{{\rm th}_\uo}}{\rho({\lambda_O}-{\lambda_I} \Bar{\gamma}_{{\rm th}_\uo})}}, \sqrt{\dfrac{\Bar{\gamma}_{{\rm th}_\ui}}{{\lambda_I} \rho}} \right\}}{\theta_{h_\ui}} \right), \label{eq:OPI}
\end{equation}
where $k_{h_\ui}$ and $\theta_{h_\ui}$ are given respectively by Eq. \eqref{eq:parameters_hi}.
\subsection{Outage Probability at Outdoor User}

In this case, the outage behavior for the $U_\uo$ occurs since $U_\uo$ cannot detect effectively its message. Hence, the OP of $U_\uo$ is defined as
\begin{IEEEeqnarray}{rCl}
    P^{U_\uo}_{\rm out} &=& {\rm Pr}(\gamma_\uo < \Bar{\gamma}_{{\rm th}_\uo}) \nonumber 
    = {\rm Pr}\left(|h_\uo| < \sqrt{\dfrac{\Bar{\gamma}_{{\rm th}_\uo}}{\rho({\lambda_O}-{\lambda_I} \Bar{\gamma}_{{\rm th}_\uo})}} \right) \nonumber
    \\
    &=& \frac{1}{\Gamma(k_{h_\uo})} \gamma\left(k_{h_\uo}, \dfrac{\sqrt{\dfrac{\Bar{\gamma}_{{\rm th}_\uo}}{\rho({\lambda_O}-{\lambda_I} \Bar{\gamma}_{{\rm th}_\uo})}}}{\theta_{h_\uo}}\right),
    \label{eq:OPO}
\end{IEEEeqnarray}
with $k_{h_\uo}$ and $\theta_{h_\uo}$ are given respectively by Eq. \eqref{eq:parameters_ho}.

\subsection{Ergodic Capacity}
In this sub-section, the EC of the $U_\ui$ and $U_\uo$ for NOMA-Aided Conventional and STAR-RIS is derived based on Theorem \ref{theo:EC}.

\begin{theorem} \label{theo:EC}
In NOMA-Aided Conventional-RIS and STAR-RIS networks, the ergodic rate for the $U_\ui$ and $U_\uo$ can be predicted as

\begin{equation}
    \begin{cases} \label{eq:ER}
        \Bar{R}_\ui \approx \log_2\left(1 +  \rho \lambda_{\ui} N_S^2 \mu_{U_\ui}^2 \pi \xi^2\beta_{U_\ui}(\beta_t + \beta_f \beta_D N_C) \right), 
        \\
        \\
        \Bar{R}_\uo  \approx 
        \log_2 \left(\hspace{-0.1cm}1 \hspace{-0.1cm}+\hspace{-0.1cm} \frac{\rho{\lambda_O} \left(  \frac{\xi N_S \mu_{U_\uo} \sqrt{\beta_{U_\uo}\pi(\beta_t+\beta_f\beta_D N_C)} + 2\sqrt{\beta_f\beta_g} N_C \mu_f \mu_g   }{2} \right)^2 }{\rho{\lambda_I} \left( \frac{\xi N_S \mu_{U_\uo} \sqrt{\beta_{U_\uo}\pi(\beta_t+\beta_f\beta_D N_C)} + 2\sqrt{\beta_f\beta_g} N_C \mu_f \mu_g }{2} \right)^2 + 1 } \hspace{-0.1cm} \right),
    \end{cases}
\end{equation}
\end{theorem}

\begin{proof}
The proof can be found in Appendix \ref{app:b}.
\end{proof}

\section{Simulation Results}\label{sec:simul}
In this section, we provide numerical results to evaluate the accuracy of the proposed analytical equations. We consider three distinct scenarios:
\underline{Sc.{\bf A}}: single- and double-reflection links are considered; \underline{Sc.{\bf B}}: only double-reflection link is considered; \underline{Sc.{\bf C}}: only single-reflection links are considered. 

We show the impact of the reflecting elements on the performance of the NOMA-aided conventional and STAR-RIS system. The MCs results are averaged over $10^5$ realizations. 
Unless stated otherwise, the adopted parameter values are presented in Table \ref{t:param}.

\begin{table}[h!] 
\caption{Adopted Simulation Parameters.}
\small
\label{t:param} 
\begin{center}
\begin{tabular}{ll}
\toprule
\bf Parameter  & \bf  Value\\
\hline 
\multicolumn{2}{c}{\bf  NOMA-aided Conventional and STAR-RIS System}\\
\hline
Transmit SNR          &   $\rho  \in [0;\, 50]$ [dB]\\
\#Total Reflecting Meta-Surfaces       &   $N\in \{50,100,200\}$\\
Split Factor & {$\eta=0.35$}\\
Power Allocation coefficients            &   ${\lambda_I}$ = {0.15}, ${\lambda_O}$ = {0.75}\\ 
Target Rate                             &   $R_1 = R_2$ = {0.01} [BPCU]\\
Distance from $S$ to $R_C$             &   {25} [m]\\
Distance from $R_S$ to $U_\ui$ and $U_\uo$          &    {5} and {20} [m] \\
Distance from $R_C$ and $R_S$              &   {15} [m]\\
Distance from $R_C$ and $U_\uo$              &   {35} [m]\\
Distance from $S$ to $R_S$              &   {35} [m]\\
\hline
\multicolumn{2}{c}{ \bf Channel Parameters}\\
\hline
$S$-$R_C$ link: \textbf{f}                  & $m_f={8}$, $\Omega_f=1$ \\
$S$-$R_S$ link: \textbf{t}                  & $m_t={1.5}$, $\Omega_t=1$ \\
$R_C$-$R_S$ link: \textbf{D}                & $m_D={8}$, $\Omega_D=1$ \\
$R_C$-$U_\uo$ link: \textbf{g}                & $m_{g}={1.8}$, $\Omega_g=1$ \\
$R_S$-$U_\ui$ link: $\textbf{u}_I$          & $m^I_u={15}$, $\Omega^I_u=1$\\
$R_S$-$U_\uo$ link: $\textbf{u}_O$                         & $m^O_u={7.5}$, $\Omega^O_u=1$  \\
Path Loss Exponent                       & ${\alpha_f = \alpha_D = 2.2}$ 
\\ & {$\alpha_{U_\ui} = \alpha_{U_{\uo}} = 2$};\,\, 
 {$\alpha_g = {2.8}$} \\
& {$\alpha_t \in\{2.8,\, 3.1,\, 3.4\}$}\\
Reference Distance & {$d_0 = 1$ [m]}\\
\toprule
 \end{tabular}
 \end{center}
\end{table}

Fig. \ref{fig:OP_eta} depicts the OP for both users {\it vs}. the split factor $\eta = \frac{N_C}{N}$ when $N=N_C + N_S=200$. We can see that when the value $N_C$ is low, $U_\uo$ has a poor outage performance due to the marginal impact of the single reflection link $S$ $\rightarrow R_C\rightarrow U_\uo$ and double-reflection link $S$ $\rightarrow R_C\rightarrow R_S \rightarrow U_\uo$ once $U_\uo$ is far from $R_S$ and $R_C$. In contrast, since the $U_\ui$ is near to the $R_S$, its performance is better when $N_S$ assumes a high value.

\begin{figure}[h]
\centering
\includegraphics[width=.82\linewidth]{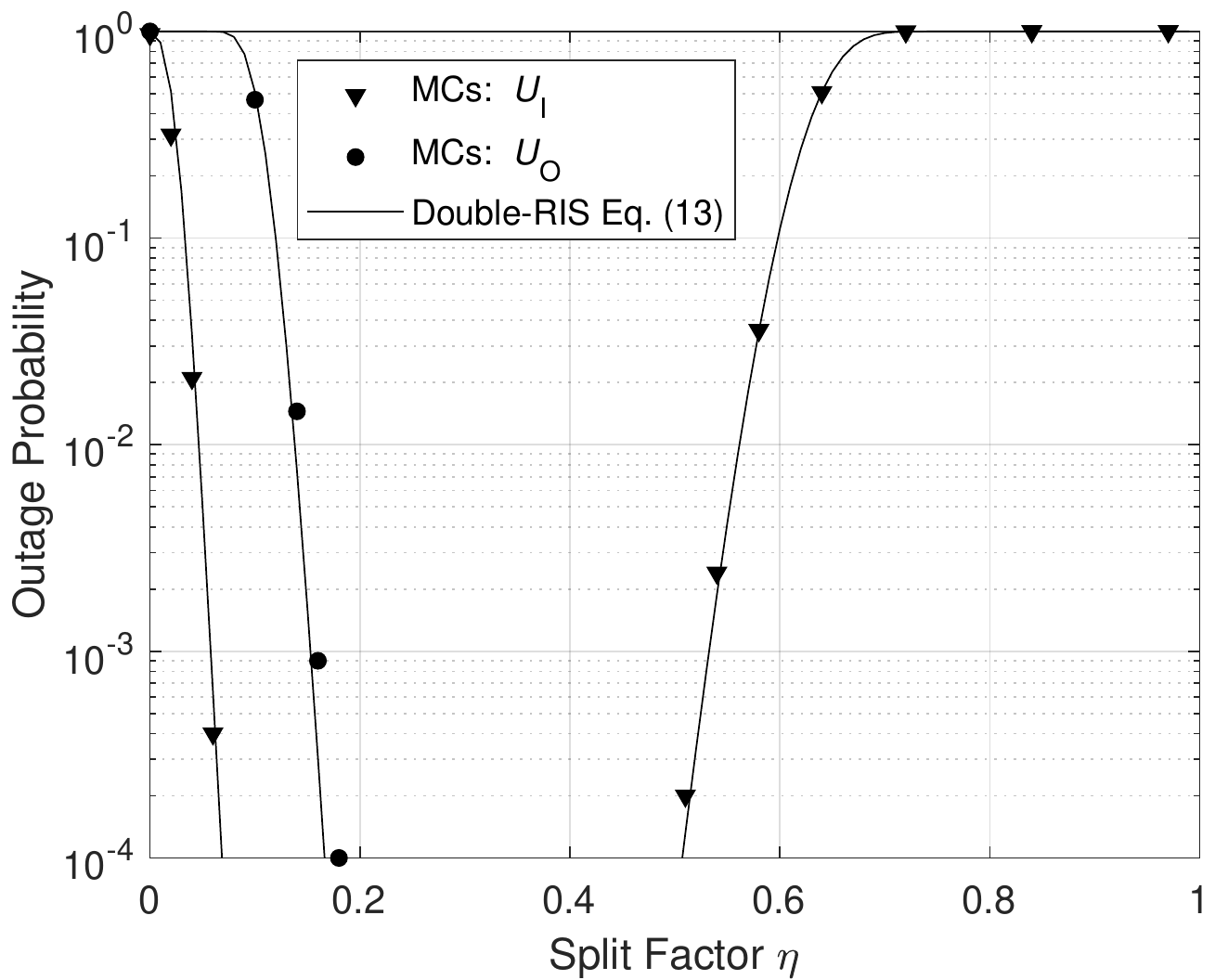}
\vspace{-4mm}
\caption{Outage performance {\it{vs}} split factor ($\eta$) for $N=200$ and 
 {$\alpha_t=3.4$} on scenario A.}
\label{fig:OP_eta}
\end{figure}

Fig. \ref{fig:EC_rho} shows the ergodic sum rate ${\text{\it{vs}}}$ the SNR ($\rho$) for three different scenarios, \textbf{A}, \textbf{B} and \textbf{C}. We can see that scenario \textbf{C} has the worse condition due to its low path-loss, besides, we can see the great potential of scenario \textbf{B} and how it can directly impact the double-RIS setup, outperforming even the single-RIS setup. Furthermore, we can see that the derived analytical equations \eqref{eq:ER} are very accurate when $\rho$ varies.

\begin{figure}[htbp!]
\centering
\includegraphics[width=.82\linewidth]{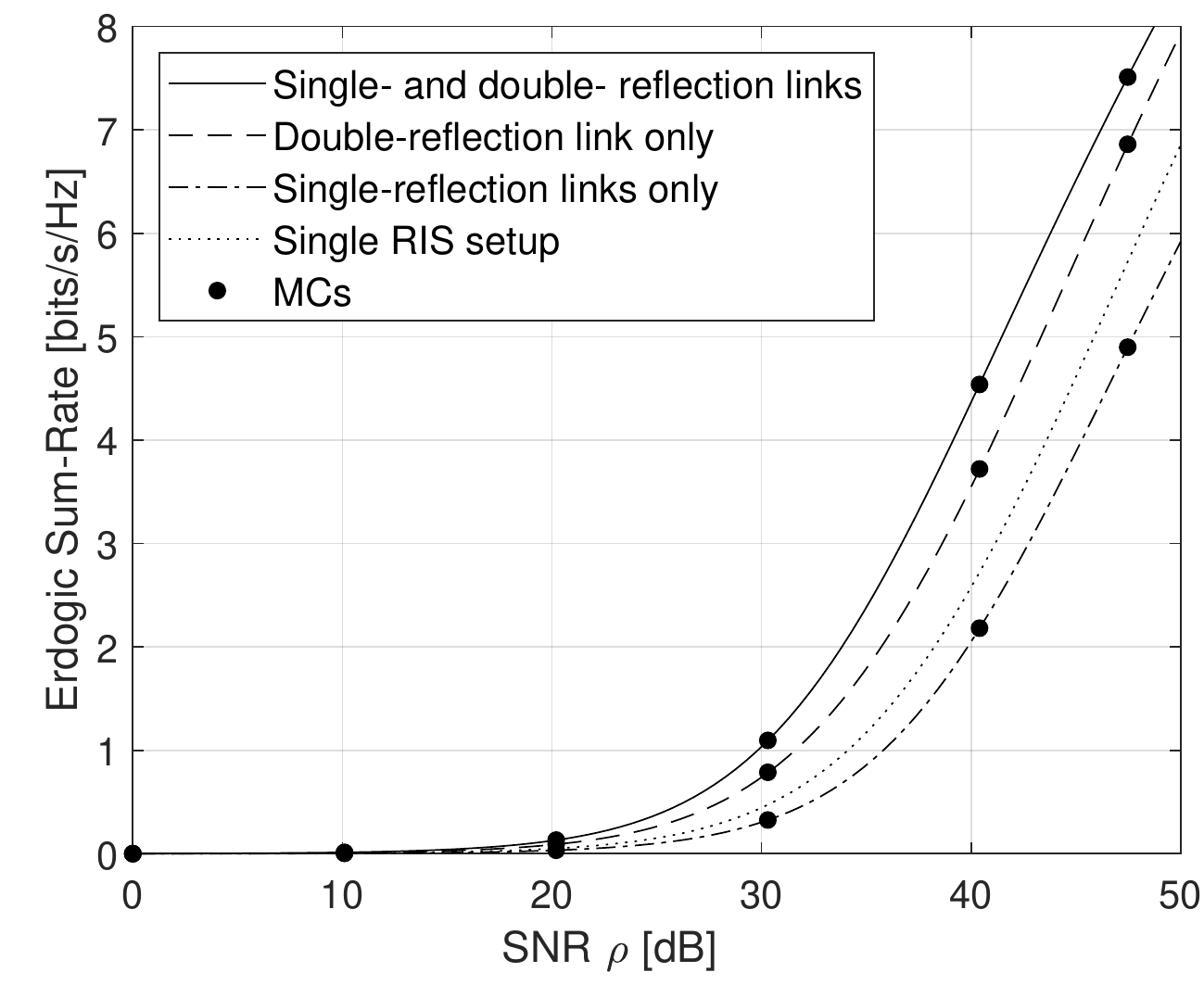}
\vspace{-5mm}
\caption{Ergodic sum-rate ($\Bar{R}_\ui + \Bar{R}_\uo$) {\it{vs}} SNR ($\rho$)  $N=200$, $\eta=0.35$ and $\alpha_t=3.4$.}
\label{fig:EC_rho}
\end{figure}

Fig. \ref{fig:EC_eta} depicts the ergodic sum rate for both users versus the split factor when $N=200$ for three different path-loss coefficients, $\alpha_t = \{2.8, 3.1, 3.4\}$, where the path-loss $\beta_t = \left(\frac{d_0}{d_t}\right)^{\alpha_t}$. We first notice that for all scenarios considered the analytical equations 
as summarized in Eq. \eqref{eq:ER} are very accurate for any value of the split factor $\eta$. Furthermore, we can see that, by improving the conditions of the direct link between the $S$ and the $R_S$, in terms of path-loss, the performance of the system can be improved in a general way, however, within these conditions, the double-RIS setup can be outperformed by the single-RIS setup. On the other hand, when $\alpha_t = \{3.1,3.4\}$, {the double-RIS setup} can obtain higher data rates than the single-RIS setup, revealing its great potential in scenarios where the direct link between the $S$ and $R_S$ is poor. Finally, we can see the scenarios where the double-RIS setup performs well, the number of meta-surface elements $N_S$ and $N_C$ can be strategically optimized aiming to obtain as better performance as possible.

\begin{figure}[htbp!]
\centering
\includegraphics[width=.82\linewidth]{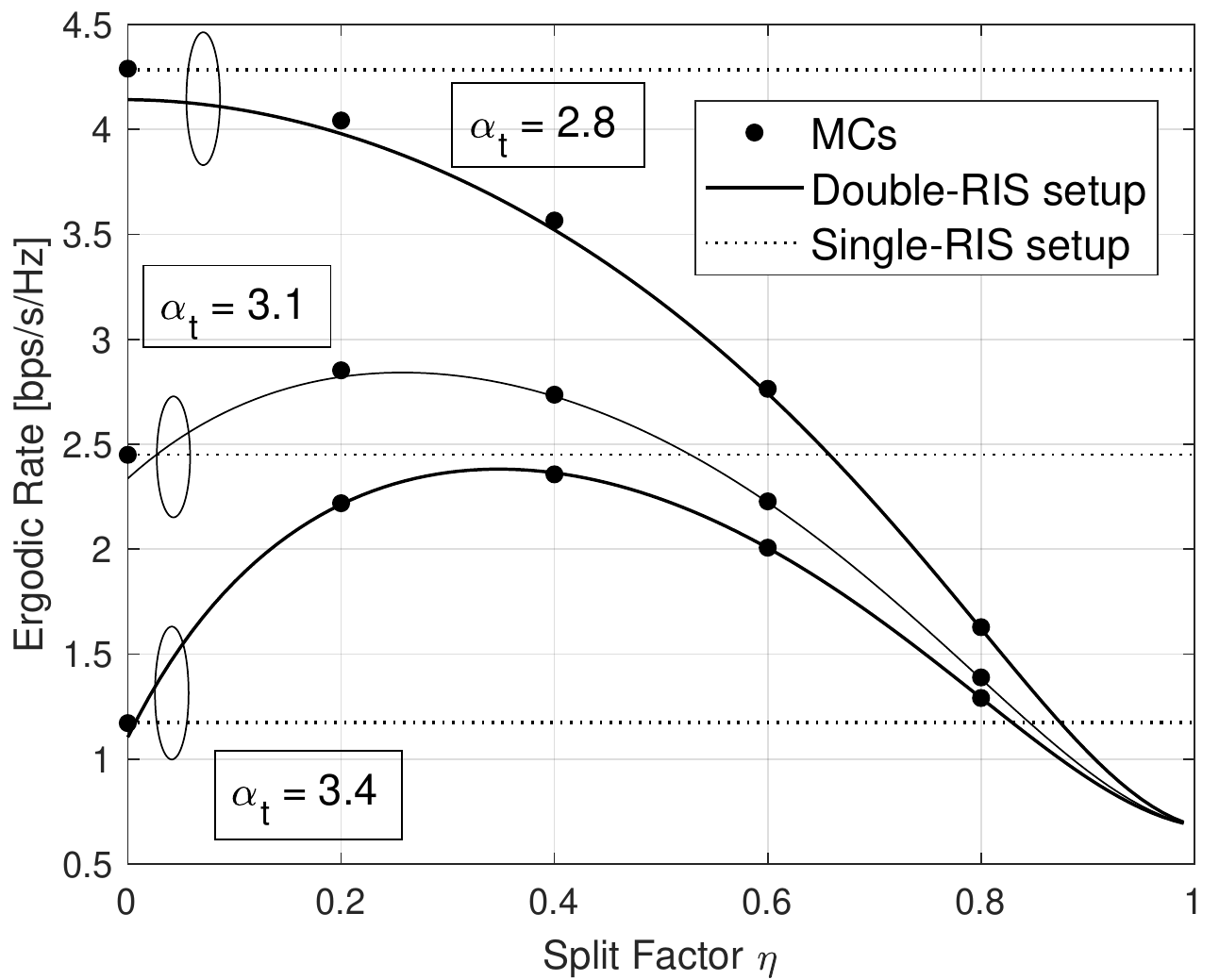}
\vspace{-6mm}
\caption{EC {\it vs} split factor ($\eta$) over an NOMA-aided conventional and STAR-RIS system {\it vs} single-RIS setup with $N=200$ and $\rho=35$ dB.}
\label{fig:EC_eta}
\end{figure}

Fig. \ref{fig:OP_rho} depicts the OP for both $U_\ui$ and $U_\uo$ also for the three different values of $N$. 
Firstly, one can 
notice that the derived analytical expressions \eqref{eq:OPI} and \eqref{eq:OPO} are very accurate for the considered parameters. Furthermore, we can notice that better performance for both $U_\ui$ and $U_\uo$ users occurs for high values of $N$. Additionally, we can observe how important is to know the CSI {aiming} to optimize the phase-shift matrix of $R_S$ and $R_C$ since it can provide very high-performance gains over the random phase-shift matrix strategy.

\begin{figure}[htbp!]
\centering
\includegraphics[width=.82\linewidth]{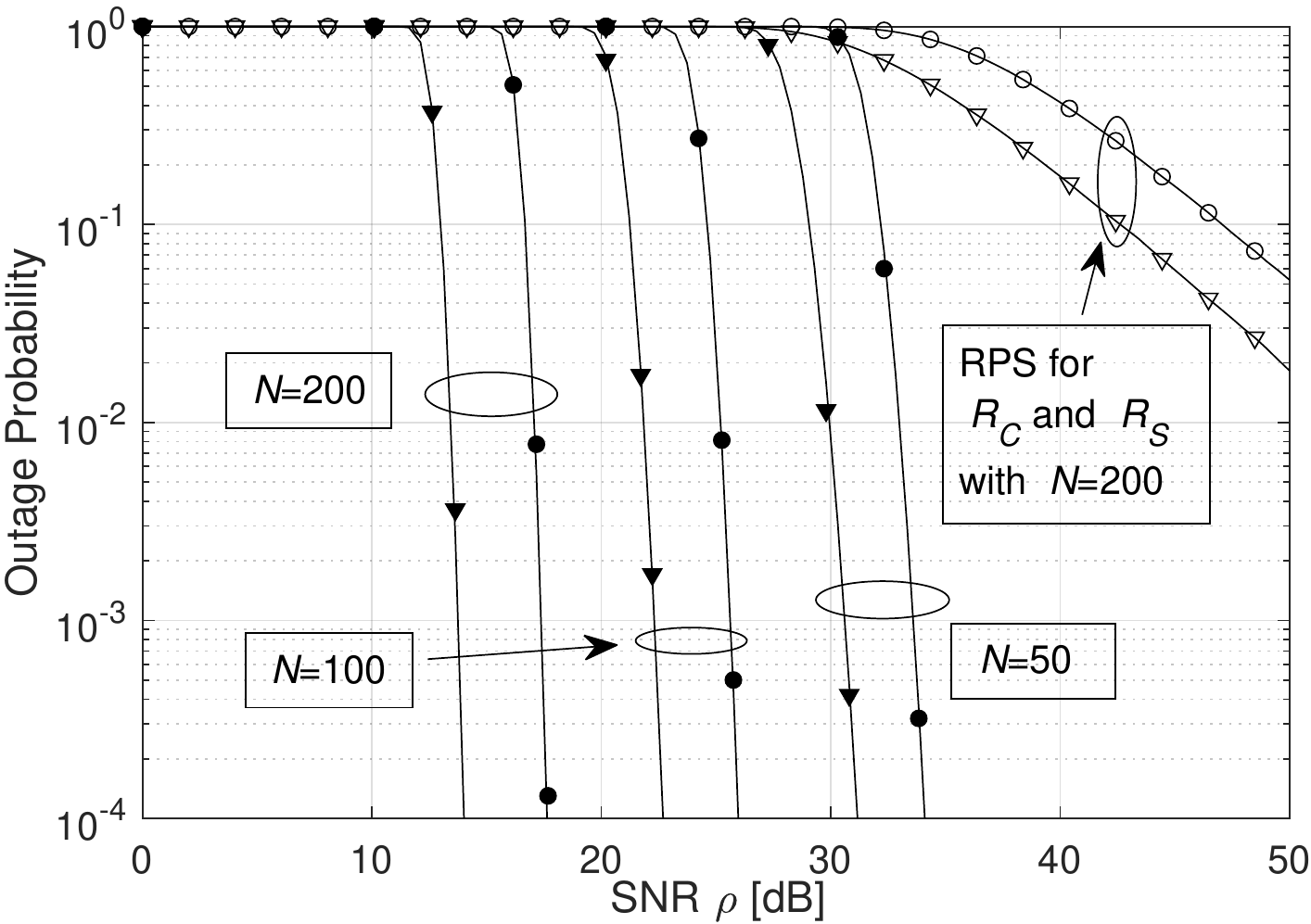}
\vspace{-6mm}
\caption{OP {\it vs} SNR ($\rho$) when $N=\{50,100,200\}$, \w{$\alpha_t=3.4$} and $\Bar{\gamma}_\ui = \Bar{\gamma}_\uo =0.5$.}
\label{fig:OP_rho}
\end{figure}

\section{Conclusions}\label{sec:concl}
In this paper, we have characterized the system performance of a NOMA-aided conventional and STAR-RIS under Nakagami-$m$ channels. We derived analytical equations for OP and EC using the Central Limit Theorem (CLT) and MM technique. The numerical (analytical and MC simulation) results reveal that the double-RIS setup can outperform the single-RIS setups when the values of $R_S$ and $R_C$ are suitably adjusted. Furthermore, in scenarios where the direct link between the $S \rightarrow R_S$ are in poor conditions w.r.t. the path-loss (large $\alpha_t$), the double-RIS can outperform the single-RIS setup. Finally, the simulation results demonstrated that the analytical proposed equations are very accurate under the parameters of the system.

\vspace{-1mm}
\section*{Acknowledgment}
This work was supported by the National Council for Scientific and Technological Development (CNPq) of Brazil under Grant 310681/2019-7, the CAPES (Financial Code 001).

\appendices

\section{Proof of Lemma \ref{lemma:Channel}}\label{app:a}
According to section \ref{ss:channelmodel} all links follow Nakagami-$m$ distribution whose mean and variance can be given as
\begin{IEEEeqnarray}{rCl}
    \mu_v &=& \dfrac{\Gamma(m_v+1/2)\sqrt{\Omega_v}}{\Gamma(m_v)\sqrt{m_v}},
    \\
    \sigma^2_v &=& \Omega_v - \mu_v^2, \qquad
    \forall v \in \{f,g,t,D,{u_\ui},{u_\uo}\},
\end{IEEEeqnarray}

The channel of $U_\uo$ in \eqref{eq:ho} can be written as 

\begin{IEEEeqnarray}{rCl}
   h_O \hspace{-0.1cm} &=& \hspace{-0.1cm}  \sqrt{\beta_{U_\uo}}\hspace{-0.1cm} \underbrace{\left(\hspace{-0.1cm}\underbrace{
    \sqrt{\beta_{t}}
    \mathbf{t} 
    +
    \sqrt{\beta_{f}\beta_{D}} \underbrace{\mathbf{f}
    \boldsymbol{\Phi}_C \mathbf{D}}_{\mathbf{w}}
     }_{\mathbf{x}}\hspace{-0.1cm}\right)\hspace{-0.1cm}
    \boldsymbol{\Phi}_S^{O} \mathbf{u}_{\uo}}_{Z_\uo} 
    + \sqrt{\beta_{f}\beta_{g}} \underbrace{\mathbf{f} \boldsymbol{\Phi}_C \mathbf{g}}_{Y},
    \nonumber
    \\
     &=& \sqrt{\beta_{U_\uo}} Z_\uo + \sqrt{\beta_f\beta_g} Y, \label{eq:appAaux1}
\end{IEEEeqnarray}
by setting $\theta_n=-\left(\angle f_n + \angle g_n \right)\; \forall n\in \mathcal{N}_C$, thus, $Y$ can be written as 
\begin{IEEEeqnarray}{rCl}
    && Y = \sum_{n=1}^{N_C} |f_n| |g_n|, \label{eq:appAaux2}
\end{IEEEeqnarray}
according to Lemma 2 of \cite{tahir2020analysis}, the distribution of $Y$ can be approximated as $Y  \overset{\rm{approx}}{\sim} \mathcal{G}\left(\dfrac{\mu^2_Y}{\sigma^2_Y},\dfrac{\sigma^2_{Y}}{\mu_{Y}}\right)$, where $\mu_Y = N_C \mu_f \mu_g$ and $\sigma^2_Y =  N_C(\Omega_f \Omega_g - \mu_f^2 \mu_g^2)$. Since $\boldsymbol{\Phi}_C$ is set to maximize the power of the single reflection link $S \rightarrow$ $R_C$ $\rightarrow$ $U_\uo$, its combination in $\mathbf{w}$ will appear random, hence, by applying Lemma 4 of \cite{tahir2020analysis} we can approximate $\mathbf{w}$ as 
\begin{equation}
    \mathbf{w} \overset{\rm{approx}}{\sim} \mathcal{CN}\left(\mathbf{0}, N_C \mathbf{I}_{N_S}\right),
\end{equation}
one can notice that \cite{ding2020impact} shows that the Central Limit Theorem (CLT) provides a good approximation for $\mathbf{w}$, however, the correlation clearly impacts the approximation of $|h_\uo|$ and it will be characterized in the numerical results in sub-section.
Besides, the distribution of $\mathbf{x}$ can be given as follows
\begin{IEEEeqnarray}{rCl}
    \mathbf{x} \overset{\rm{approx}}{\sim} \mathcal{CN}\left(\mathbf{0},(\beta_t +\beta_f\beta_D N_C) \mathbf{I}_{N_S}\right),
\end{IEEEeqnarray}
since the distribution of $x_n$ follows complex Gaussian, the distribution of $|x_n|$, follows Rayleigh 
 $\forall n \in \mathcal{N}_S$ whose mean is given as $\mu_{x} \triangleq \mathbb{E}[|x_n|] = \sqrt{\frac{\pi(\beta_t+\beta_f\beta_D N_C)}{4}}$ and variance can be given as $ \sigma^2_{x} \triangleq \mathbb{V}[|x_n|] = \frac{(4-\pi)(\beta_t + \beta_f\beta_DN_C)}{4}$.

Furthermore we set $\varphi_n^\uo = -\left(\angle x_n + \angle u_{\uo,n}\right)\; \forall n\in \mathcal{N}_S$, in order to maximize the received power of $U_\uo$ thus $Z_\uo$ can be written as follows
\begin{IEEEeqnarray}{rCl}
&& Z_\uo = \xi \sum_{n=1}^{N_S} |x_n| |u_{\uo,n}|, \label{eq:appAaux3}
\end{IEEEeqnarray}
we should notice that differently of \eqref{eq:appAaux2}, eq. \eqref{eq:appAaux3} represents the sum of the product of a Rayleigh and Nakagami-$m$ r.v.; however, since Rayleigh can be viewed as a particular case of a Nakagami-$m$ r.v., we also applied Lemma 2 of \cite{tahir2020analysis} to approximate the PDF of $Z_{\uo}$ as Gamma r.v., hence, $Z_\uo \overset{\rm{approx}}{\sim} \mathcal{G}\left(\dfrac{\mu^2_{Z_\uo}}{\sigma^2_{Z_\uo}},\xi \dfrac{\sigma^2_{{Z_\uo}}}{\mu_{{Z_\uo}}}\right)$, where $\mu_{Z_\uo} = \xi N_S \mu_x \mu_{U_\uo}$ and $\sigma^2_{Z_\uo} \approx \xi^2 N_S(\Omega_{U_\uo}(\beta_t +\beta_f\beta_D N_C) -  \mu^2_{U_\uo}\mu_{x}^2)$, where the {\it approx.} is due to the correlation in $\mathbf{x}$.

Hence, the distribution of $|h_O|$ in \eqref{eq:appAaux1}, is the sum of two scaled gamma random variables with different shapes and scales parameters, since the exact closed-form expression for its PDF can be toilsome to obtain, we approximate the distribution of $|h_O|$ according to the well-known \textit{Welch-Satterthwaite} \cite{abusabah2020approximate} as follows 
\vspace{-0.2cm}
\begin{IEEEeqnarray}{rCl}
\hspace{-4mm} |h_O| \sim \mathcal{G}\left(\frac{\left( \sqrt{ \beta_{U_\uo}}  \xi N_S \mu_x \mu_{U_\uo} + \sqrt{\beta_f\beta_g} N_C \mu_f \mu_g \right)^2}{\xi^2 \beta_{U_\uo} \sigma^2_{Z_O} + \beta_f\beta_g \sigma^2_Y }, 
    \right. \nonumber
    \\  
  \hspace{-4mm}  \left.
   \dfrac{\xi^2 \beta_{U_\uo} \sigma^2_{Z_O} + \beta_f\beta_g \sigma^2_Y }{ \xi \sqrt{ \beta_{U_\uo}} N_S \mu_x \mu_{U_\uo} + \sqrt{\beta_f\beta_g} N_C \mu_f \mu_g } \right).
\end{IEEEeqnarray}
By following the same step in \eqref{eq:appAaux1} and applying optimal phase-shift in $\varphi^\ui_n$, the channel of $U_\ui$ in \eqref{eq:hi} can be written as
\begin{IEEEeqnarray}{rCl}
    |h_I| &=& |\sqrt{\beta_{U_\ui}} \underbrace{\mathbf{x}
    \boldsymbol{\Phi}_S^{I} \mathbf{u}_{\ui}}_{Z_\ui}|  =  \sqrt{\beta_{U_\ui}} Z_I,
\end{IEEEeqnarray}
following the same steps done above, we obtain that $Z_\ui \overset{\rm{approx}}{\sim} \mathcal{G}\left(\dfrac{\mu^2_{Z_\ui}}{\sigma^2_{Z_\ui}}, \dfrac{\sigma^2_{{Z_\ui}}}{\mu_{{Z_\ui}}}\right)$, where the mean and variance are given as $\mu_{Z_\ui} = \xi N_S \mu_x \mu_{U_\ui}$ and $\sigma^2_{Z_\ui} \approx \xi^2  N_S( \Omega_{U_\ui}(\beta_t + \beta_f \beta_D N_C) - \mu^2_{U_\ui}\mu_{x}^2)$ respectively, therefore 
\begin{equation}
|h_\ui|\sim\mathcal{G}\left(\dfrac{\mu^2_{Z_\ui}}{\sigma^2_{Z_\ui}},\sqrt{\beta_{U_\ui}} \dfrac{\sigma^2_{{Z_\ui}}}{\mu_{{Z_\ui}}}\right).
\end{equation}

\vspace{-2mm}
\section{Proof of Theorem \ref{theo:EC}}\label{app:b}
The EC of $U_\ui$ and $U_\uo$ can be written as
\begin{IEEEeqnarray}{rCl}
    \Bar{R}_i &=& \mathbb{E}\left[\log_2\left(1 +\gamma_i\right)\right], \qquad \forall i\in\{\ui,\uo\},
\end{IEEEeqnarray}

\noindent where $\gamma_\ui$ and $\gamma_\uo$ are given respectively by \eqref{eq:gammai} and \eqref{eq:gammao}. By utilizing Jensen's inequality we can obtain an upper bound based on 
$\log_2(1+x)$, which is concave w.r.t. $x$, therefore
\begin{IEEEeqnarray}{rCl}
    \Bar{R}_\ui \leq \log_2\left(1+\rho {\lambda_I} \mathbb{E}[ |h_\ui|^2]\right), \label{eq:rateupi}\\
    \Bar{R}_\uo \leq \log_2\left(1 + \frac{\rho{\lambda_O}\mathbb{E}[ |h_\uo|^2]}{\rho{\lambda_I}\mathbb{E}[ |h_\uo|^2]+1}\right). \label{eq:rateupo}
\end{IEEEeqnarray}

Since $|h_i|\sim\mathcal{G}(k_{h_i},\theta_{h_i})$, we have that
\begin{IEEEeqnarray}{rCl}
    f_{|h_i|^2}(x) &=& f_{|h_i|}(x)\frac{1}{2\sqrt{x}} 
    = \frac{1}{2 \theta_i^{k_i} \Gamma(k_i)} x^{\frac{k_i}{2}-1} e^{-\frac{x^{1/2}}{\theta_i}},
\end{IEEEeqnarray}
therefore $|h_i|^2\sim\mathcal{G}\mathcal{G}(k_i,\theta_i^2,0,1/2)$ \cite[Eq. 17.116]{johnson1994continuous}, thus, the mean of generalized gamma r.v. is given by \cite[Eq. 17.118]{johnson1994continuous}
\begin{equation} \label{eq:meangg}
    \mathbb{E}[|h_i|^2] = \theta_i^2 \frac{\Gamma(k_i+2)}{\Gamma(k_i)},
\end{equation}
One can notice that if $k_i\rightarrow +\infty$, $\Gamma\left(k_i+2\right) \approx \Gamma\left(k_i\right) k_i^{2}$, thus \eqref{eq:meangg} can be approximate as
\vspace{-0.1cm}
\begin{equation} \label{eq:meanapprox}
    \mathbb{E}[|h_i|^2] \approx \theta_i^2 k_i^2,
\end{equation}
thus substituting \eqref{eq:meanapprox} into \eqref{eq:rateupi} and \eqref{eq:rateupo} respectively, \eqref{eq:ER} are obtained. This completes the proof.

\vspace{-0.1cm}
\bibliographystyle{IEEEtran}
\bibliography{Ref_file.bib}

\end{document}